\documentclass[aps,%
twocolumn,%
floatfix,%
reprint,%
superscriptaddress,%
longbibliography]{revtex4-1}%
\usepackage[utf8]{inputenc}
\usepackage{mathrsfs}       
\usepackage{mathtools}
\usepackage{amsmath}
\usepackage{amsthm}         
\usepackage{amsfonts}       
\usepackage[titletoc,title]{appendix}
\usepackage{algorithm,algorithmic}
\usepackage{bbm}            
\usepackage{bm}             
\usepackage{times}

\usepackage[colorlinks=true,%
      linkcolor=beamer@blendedblue,%
      bookmarks=true,%
      breaklinks=true,%
      filecolor=beamer@blendedblue,%
      anchorcolor=yellow,%
      citecolor=beamer@blendedblue,%
      urlcolor=beamer@blendedblue,%
      pdfauthor={Kun Wang, Yu-Ao Chen, and Xin Wang},%
      pdfsubject={Measurement Error Mitigation via Truncated Neumann Series},%
      CJKbookmarks=true]{hyperref}

\usepackage[dvipsnames]{xcolor}
\definecolor{beamer@blendedblue}{rgb}{0.2,0.2,0.7}

\newtheorem{definition}{Definition}
\newtheorem{proposition}[definition]{Proposition}
\newtheorem{lemma}[definition]{Lemma}
\newtheorem{theorem}[definition]{Theorem}

\mathchardef\ordinarycolon\mathcode`\:
\mathcode`\:=\string"8000
\def\vcentcolon{\mathrel{\mathop\ordinarycolon}}
\begingroup \catcode`\:=\active
  \lowercase{\endgroup
  \let :\vcentcolon
  }

\DeclareFontFamily{U}{mathx}{\hyphenchar\font45}
\DeclareFontShape{U}{mathx}{m}{n}{<-> mathx10}{}
\DeclareSymbolFont{mathx}{U}{mathx}{m}{n}
\DeclareMathAccent{\widebar}{0}{mathx}{"73}

\newcommand{\ket}[1]{\left\vert{#1}\right\rangle}
\newcommand{\bra}[1]{\left\langle{#1}\right\vert}

\newcommand\proj[1]{\vert{#1}\rangle\!\langle{#1}\vert}
\newcommand{\opn}[1]{\operatorname{#1}}
\DeclareMathOperator{\tr}{Tr}  
\newcommand{\1}{\mathbbm{1}}

\newcommand{\norm}[2]{\ensuremath{\left\lVert#1\right\rVert_{#2}}}%

\newcommand*{\cO}{\mathcal{O}}

\newcommand{\bE}{\mathbb{E}}

\begin{document}

\title{\Large\textbf{Measurement Error Mitigation via Truncated Neumann Series}}

\author{Kun Wang}
\email{wangkun28@baidu.com}
\affiliation{Institute for Quantum Computing, Baidu Research, Beijing 100193, China}

\author{Yu-Ao Chen}
\email{chenyuao@baidu.com}
\affiliation{Institute for Quantum Computing, Baidu Research, Beijing 100193, China}

\author{Xin Wang}
\email{wangxin73@baidu.com}
\affiliation{Institute for Quantum Computing, Baidu Research, Beijing 100193, China}

\begin{abstract}
Measurements on near-term quantum processors are inevitably subject to hardware imperfections that lead to readout errors. Mitigation of such unavoidable errors is crucial to better explore and extend the power of near-term quantum hardware. In this work, we propose a method to mitigate measurement errors in computing quantum expectation values using the truncated Neumann series. The essential idea is to cancel the errors by combining various noisy expectation values generated by sequential measurements determined by terms in the truncated series.
We numerically test this method and find that the computation accuracy is substantially improved. Our method possesses several advantages:
it does not assume any noise structure, it does not require the calibration procedure to learn the noise matrix {a prior}, and most importantly, the incurred error mitigation overhead is independent of system size, as long as the noise resistance of the measurement device is moderate. All these advantages empower our method as a practical measurement error mitigation method for near-term quantum devices.
\end{abstract}

\date{\today}
\maketitle

\section{Introduction}

Quantum computers hold great promise for a variety of
scientific and industrial applications~\cite{mcardle2020quantum,cerezo2020variational,bharti2021noisy,endo2021hybrid}. However, in the current stage
noisy intermediate-scale quantum (NISQ) computers~\cite{preskill2018quantum}
introduce significant errors that must be dealt with before
performing any practically valuable tasks.
Errors in a quantum computer are typically classified into quantum gate errors
and measurement errors.
For quantum gate errors, various quantum error mitigation techniques
have been proposed to mitigate the damages caused by errors
on near-term quantum devices~\cite{temme2017error,endo2018practical,li2017efficient,mcclean2017hybrid,McClean2020Decoding,mcardle2019error,bonet2018low,he2020resource,giurgica2020digital,kandala2019error,endo2021hybrid,sun2021mitigating,czarnik2020error}.
For measurement errors, experimental works have demonstrated that
measurement errors in quantum devices can be well understood in terms of
classical noise models~\cite{chow2012universal,kandala2019error,chen2019detector},
which is recently rigorously justified~\cite{geller2020rigorous}.
Specifically, a $n$-qubit noisy measurement device
can be characterized by a noise matrix $A$ of size $2^n\times 2^n$.
The element in the $\bm{x}$-th row and $\bm{y}$-th column,  $A_{\bm{x}\bm{y}}$,
is the probability of obtaining a outcome $\bm{x}$ provided that the true outcome is $\bm{y}$.
If one has access to this stochastic matrix, it is
straightforward to classically reverse the noise effects
simply by multiplying the probability vector obtained from experimental statistics
by this matrix's inversion.
However, there are several limitations of this matrix inversion approach:
(i) The complete characterization of $A$ requires $2^n$ calibration experiment setups and thus is not scalable.
(ii) The matrix $A$ may be singular for large $n$, preventing direct inversion. (iii) The
inverse $A^{-1}$ is hard to compute and might not be a stochastic matrix, indicating
that it can produce negative probabilities.

Several approaches have been proposed to deal with these
issues~\cite{maciejewski2020mitigation,tannu2019mitigating,nachman2019unfolding,hicks2021readout,bravyi2020mitigating,geller2020efficient,murali2020software,kwon2020hybrid,funcke2020measurement,zheng2020bayesian,maciejewski2021modeling,barron2020measurement}.
For example, Ref.~\cite{chen2019detector,maciejewski2020mitigation} elucidated that
the quality of the measurement calibration and the number of measurement samples
affected the performance of measurement error mitigation methods dramatically.
Motivated by the unfolding algorithms in high energy physic,
Ref.~\cite{nachman2019unfolding,hicks2021readout} used the
iterative Bayesian unfolding approach to avoid pathologies from the matrix inversion.
Ref.~\cite{bravyi2020mitigating} introduced a new classical noise model based on the
continuous time Markov processes and proposed an error mitigation approach
that cancels errors using the quasiprobability decomposition technique~\cite{pashayan2015estimating,temme2017error,howard2017application,endo2018practical,takagi2020optimal,jiang2020physical,regula2021operational}.
However, most of these works make an explicit assumption on the physical noise model
and require the calibration procedure to learn the stochastic matrix $A$,
and thus is not scalable in general.
Recently, Ref.~\cite{berg2020model} proposed a noise model-free measurement error mitigation method that
forces the bias in the expectation value to appear as a multiplicative factor that can be removed.

In this work, we propose a measurement error mitigation method motivated by the Neumann series,
applicable for any quantum algorithms where the measurement statistics are used for computing the
expectation values of observables.
The idea behind this method is to cancel the measurement errors by utilizing the noisy expectation values generated by sequential measurements,
each determined by a term in the truncated Neumann series.
The method is deliberately simple, does not make any assumption about the actual physical noise model,
and does not require calibrating the stochastic matrix {a priori}.

The paper is organized as follows. Section~\ref{sec:expectation value}
describes the quantum task of computing expectation values and
explains how the noisy measurement incurred bias to the results.
Section~\ref{sec:Neumann series} presents the error mitigation technique
via truncated Neumann series.
Section~\ref{sec:experimental results} reports the
experimental demonstration of our error mitigation method.
The Appendices summarize technical details used in the main text.

\section{Computing the expectation value}\label{sec:expectation value}

Let $\rho$ be an $n$-qubit quantum state generated by a quantum circuit.
Most of the quantum computing tasks end with computing the expectation value $\tr[O\rho]$ of a given observable $O$ within
a prefixed precision $\varepsilon$, by post-processing the measurement outcomes of the quantum state.
This task is the essential component of multifarious quantum algorithms, 
notable practical examples of which are variational quantum eigensolvers~\cite{peruzzo2014variational,mcclean2016theory},
quantum approximate optimization algorithm~\cite{farhi2014quantum}, and
quantum machine learning~\cite{biamonte2017quantum,havlivcek2019supervised}.

For simplicity, we assume that the observable $O$ is diagonal
in the computational basis and its elements take values in the range $[-1,1]$, i.e.,
\begin{align}\label{eq:observable}
  O = \sum_{\bm{x}\in\{0,1\}^n} O(\bm{x})\proj{\bm{x}},\quad
  \vert O(\bm{x})\vert \leq 1,
\end{align}
where $O(\bm{x})$ is the $\bm{x}$-th diagonal element of $O$
and $\vert\alpha\vert$ is the absolute value of $\alpha$.
Note that we adopt the convention that the diagonal elements are indexed from $0$.
Consider $M$ independent experiments where in each round
we prepare the state $\rho$ using the same quantum circuit
and measure each qubit in the computational basis (see, e.g., Fig.~\ref{fig:expectation}). Let
$\bm{s}^m\in\{0,1\}^n$ be the measurement outcome observed in the $m$-th round.
We further define the empirical mean value
\begin{align}\label{eq:ideal-average}
\eta^{(0)} := \frac{1}{M}\sum_{m=1}^M O(\bm{s}^m).
\end{align}
Let $\opn{vec}(\rho)$ be the $2^n$-dimensional column vector formed by the diagonal elements of $\rho$.
Then~\cite{bravyi2020mitigating}
\begin{align}\label{eq:ideal-expectation}
  E^{(0)} := \bE[\eta^{(0)}]
= \sum_{\bm{x}\in\{0,1\}^n}O(\bm{x})\langle \bm{x}\vert\opn{vec}(\rho) = \tr[O\rho],
\end{align}
where $\bE[X]$ is the expectation of the random variable $X$.
Eq.~\eqref{eq:ideal-expectation} implies that $\eta^{(0)}$ is an unbiased estimator of $\tr[O\rho]$.
What's more, the standard deviation $\sigma(\eta^{(0)})\leq1/\sqrt{M}$.
By Hoeffding's inequality~\cite{hoeffding1994probability}, $M=2\log(2/\delta)/\varepsilon^2$ would guarantee that
\begin{align}
    \opn{Pr}\{\vert\eta^{(0)} - \tr[O\rho]\vert \leq \varepsilon\} \geq 1 - \delta,
\end{align}
where $\opn{Pr}\{\cdot\}$ is the event's probability, $\delta$ is the specified confidence,
and all logarithms are in base $2$ throughout this paper.

\begin{figure}
  \centering
  \includegraphics[width=0.35\textwidth]{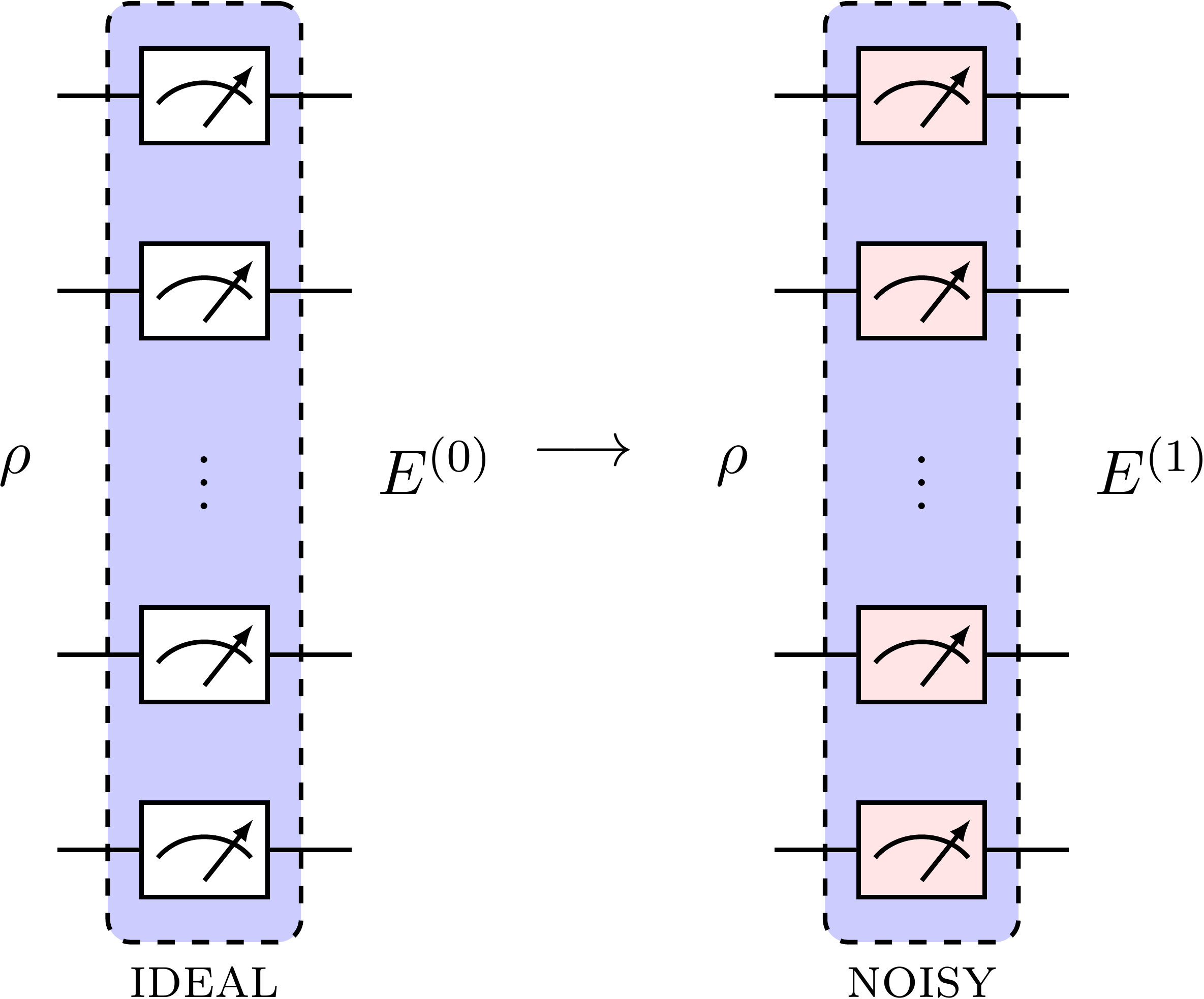}
  \caption{Computing the expectation value $\tr[O\rho]$ with the ideal
        measurement device (left) and the noisy measurement device (right)~\cite{note}.}
  \label{fig:expectation}
\end{figure}

However, measurement devices on current quantum hardware inevitably suffer from hardware
imperfections that lead to readout errors, which are manifested as a bias toward
the expectation values we aim to compute (cf. the right side of Fig.~\ref{fig:expectation}).
As previously mentioned, in the most general scenario, these errors are modeled by a
$2^n\times 2^n$ noise matrix $A$.
If there were no measurement error at all, $A$ is the identity matrix $I$.
The off-diagonal elements of $A$ completely characterize the readout errors.
By definition, $A$ is column-stochastic in the sense that
the elements of each column are non-negative and sum to $1$.

Suppose now that we adopt the same procedure for computing $\eta$ in~\eqref{eq:ideal-average},
where we perform $M$ independent experiments and collect the measurement outcomes.
Denote by $\bm{s}^{m,1}\in\{0,1\}^n$ the outcome observed in the $m$-th round,
where the superscript $1$ indicates that the noisy measurement is applied.
As~\eqref{eq:ideal-average}, we define
\begin{align}\label{eq:noisy-average}
\eta^{(1)} := \frac{1}{M}\sum_{m=1}^M O(\bm{s}^{m,1}).
\end{align}
We prove in Appendix~\ref{appx:noisy-expectation} that
\begin{align}\label{eq:noisy-expectation}
E^{(1)} := \bE[\eta^{(1)}] = \sum_{\bm{x}\in\{0,1\}^n}O(\bm{x})\langle \bm{x}\vert A \opn{vec}(\rho),
\end{align}
indicating that $\eta^{(1)}$ is no longer an estimator of $\tr[O\rho]$.
Comparing Eqs.~\eqref{eq:ideal-expectation} and~\eqref{eq:noisy-expectation},
we find that in the ideal case, the sampled probability distribution approximates $\opn{vec}(\rho)$
due to the weak law of large numbers,
while in the noisy case, the sampled probability distribution approximates $A\opn{vec}(\rho)$,
leading to a bias in the estimator.

\section{Error mitigation via truncated Neumann series}\label{sec:Neumann series}

A direct approach to eliminate the measurement errors from $A\opn{vec}(\rho)$ is
to apply the inverse matrix $A^{-1}$. However, this approach is resource-consuming
and only feasible when $n$ is small. To deal with this difficulty,
we simulate the effect of $A^{-1}$ using a truncated Neumann series. That is,
$A^{-1}$ is approximated by a linear combination of the terms $A^k$ for different $k$,
with carefully chosen coefficients.
This idea has previously been applied for linear data detection
in massive multiuser multiple-input multiple-output wireless systems~\cite{wu2013approximate}.

Define the \emph{noise resistance} of the noise matrix $A$ as
\begin{align}\label{eq:noise resistance}
   \xi := 2\left(1 - \min_{\bm{x}\in\{0,1\}^n}\bra{\bm{x}}A\ket{\bm{x}}\right).
\end{align}
By definition, $1-\xi/2$ is the minimal diagonal element of $A$.
Intuitively, $\xi/2$ characterizes the noisy measurement device's behavior
in the worst-case scenario since it is the maximal probability for which the
true outcome should be $\bm{x}$ yet the actual outcome is not $\bm{x}$.
In the following, we assume $\xi<1$, which is equivalent to the condition that
the minimal diagonal element of $A$ is larger than $0.5$.
This assumption is reasonable since otherwise
the measurement device is too noisy to be applied from the practical perspective.
Under this assumption, the stochastic matrix $A$ is nonsingular
and the Neumann series implies that~\cite[Theorem 4.20]{stewart1998matrix}
\begin{subequations}\label{eq:Neumann-series}
\begin{align}
A^{-1}
&=\sum_{k=0}^\infty (I-A)^k \\
&=\sum_{k=0}^K (I-A)^k + \cO((I-A)^{K+1}) \\
&= \sum_{k=0}^K c_K(k) A^k + \cO((I-A)^{K+1}),
\end{align}
\end{subequations}
where for arbitrary non-negative integers $0\leq k \leq K$,
the coefficient function is defined as
\begin{align}\label{eq:c_k}
    c_K(k) := (-1)^{k}\binom{K+1}{k+1},
\end{align}
and $\binom{n}{k}$ is the binomial coefficient.
Intuitively, Eq.~\eqref{eq:Neumann-series} indicates that one may approximate
the inverse matrix $A^{-1}$ using the first $K$ Neumann series terms,
if the behavior of the remaining terms $\cO((I-A)^{K+1})$ can be bounded.
We show that this is indeed the case in the measurement error mitigation task.
More specifically, using the first $K+1$ terms in the expansion~\eqref{eq:Neumann-series} of $A^{-1}$,
we obtain the following.
\begin{theorem}\label{thm:approximation}
For arbitrary positive integer $K$, it holds that
\begin{align}\label{eq:approximation}
    \left\vert \tr[O\rho] - \sum_{k=1}^{K+1}c_K(k-1) E^{(k)}\right\vert
\leq \xi^{K+1},
\end{align}
where
\begin{align}\label{eq:E_k}
    E^{(k)} := \sum_{\bm{x}\in\{0,1\}^n}O(\bm{x})\langle \bm{x}\vert A^{k}\opn{vec}(\rho).
\end{align}
\end{theorem}
The proof is given in Appendix~\ref{appx:approximation}.
As evident from Theorem~\ref{thm:approximation},
the noise resistance $\xi$ of the noise matrix $A$ determines the number of terms required in the
truncated Neumann series to approximate $A^{-1}$ to the desired precision.
What is more, since $\xi<1$, the approximation error decays exponentially in terms of $K$.
By the virtue of~\eqref{eq:noisy-expectation}, $E^{(k)}$ can be viewed as the noisy
expectation value generated by a noisy measurement device whose corresponding noise
matrix is $A^k$.
Let $\overline{E}:=\sum_{k=1}^{K+1}c_K(k-1) E^{(k)}$.
Theorem~\ref{thm:approximation} inspires a systematic way to estimate the expectation
value $\tr[O\rho]$ in two steps.

\begin{figure}[!htbp]
  \centering
  \includegraphics[width=0.4\textwidth]{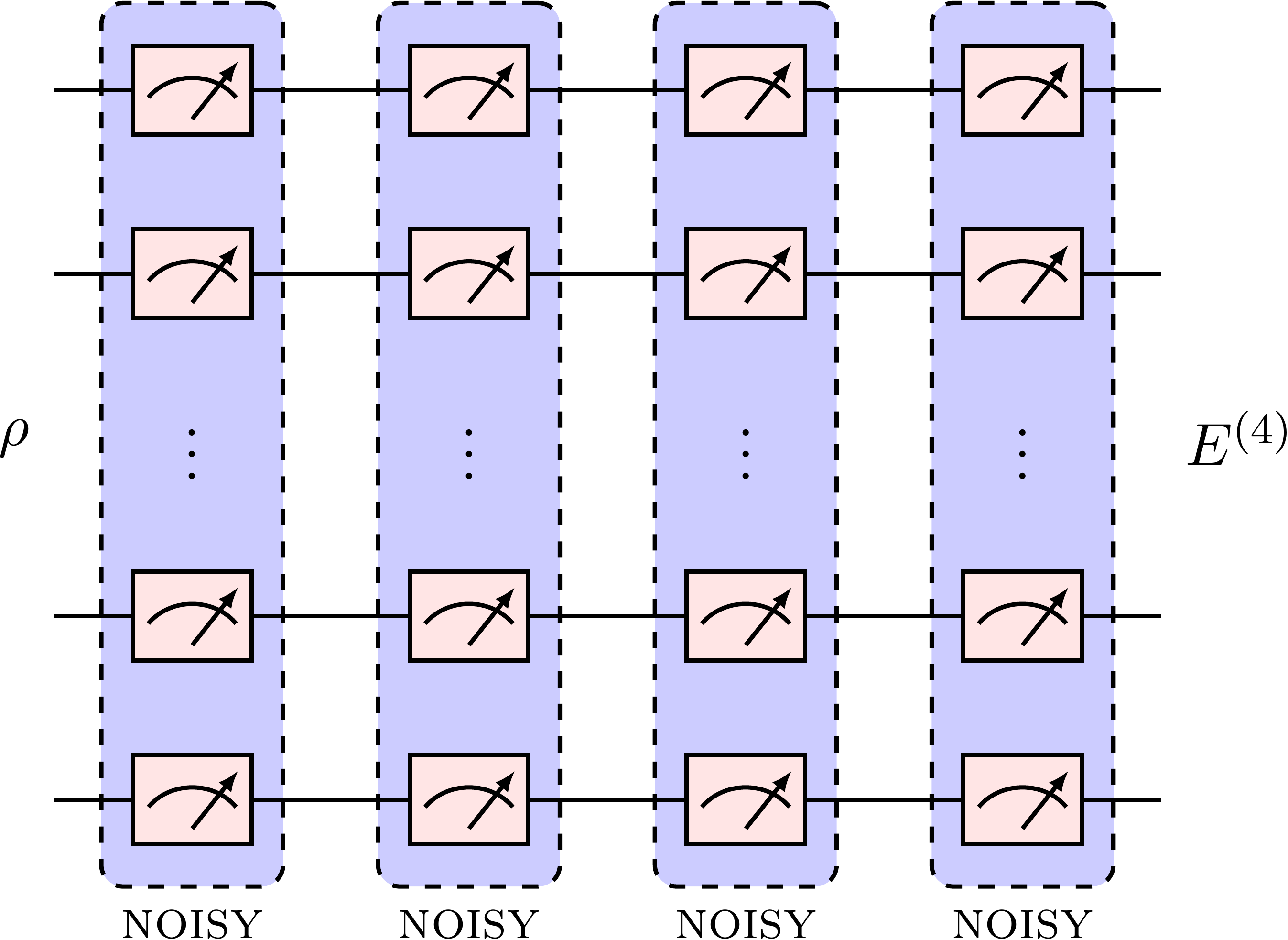}
  \caption{Experimental setup for estimating $E^{(4)}$,
        in which the noisy measurement device (box in blue) is executed $4$ times sequentially.}
  \label{fig:expectation-k4}
\end{figure}

Firstly, we choose $K$ for which
the RHS. of~\eqref{eq:approximation} evaluates to $\varepsilon$, yielding
the optimal truncated number
\begin{align}\label{eq:K-opt}
    K = \left\lceil\frac{\log\varepsilon}{\log\xi} - 1\right\rceil.
\end{align}
Such a choice guarantees that $\overline{E}$ is $\varepsilon$-close to the expectation value $\tr[O\rho]$.
Secondly, we compute the quantity $\overline{E}$
by estimating each term $E^{(k)}$ and computing the linear combination according
to the coefficients $c_K$.
Since $\overline{E}$ itself is only an $\varepsilon$-estimate
of $\tr[O\rho]$, it suffices to
approximate $\overline{E}$ within an error $\varepsilon$.
Motivated by the relation between $\eta^{(1)}$ and $E^{(1)}$ (see the discussions and
calculations in obtaining~\eqref{eq:noisy-average}),
we declare that each $E^{(k)}$ can be estimated via the following procedure:
\begin{enumerate}
  \item Generate a quantum state $\rho$.
  \item Using $\rho$ as input, execute the noisy measurement device $k$ times
        \emph{sequentially} and collect the outcome produced by the final measurement device,
        i.e., the $k$-th measurement device.
  \item Repeat the above two steps $M$ rounds and collect the measurement outcomes.
  \item Define an average analogous to~\eqref{eq:noisy-average} and output it as an estimate
        of $E^{(k)}$.
\end{enumerate}
We elaborate thoroughly on the concept of sequential measurement
in Appendix~\ref{appx:sequential measurements} and show that
the classical noise model describing the
sequential measurement repeating $k$ times is effectively characterized by the
stochastic matrix $A^k$.
For a sequential measurement repeating $k$ times,
one can think of the rightmost $k-1$ measurements as
implementing the calibration subroutine since they accept the computational basis
states as inputs. To some extent, this is a \emph{dynamic} calibration procedure where
we do not statically enumerate all computational bases as input states
but dynamically prepare the input states based on the output information
of the target state from the first measurement device.
For illustrative purpose,
we demonstrate in Fig.~\ref{fig:expectation-k4} the experimental setup for estimating
the noisy expectation value $E^{(4)}$, where the measurement device is repeated four
times in each round. We summarize the whole procedure in the following Algorithm~\ref{alg:Neumann series}.

\renewcommand{\algorithmicrequire}{\textbf{Input:}}
\renewcommand{\algorithmicensure}{\textbf{Output:}}
\begin{algorithm}[H]
\caption{Error mitigation via truncated Neumann series}
\begin{algorithmic}[1] \label{alg:Neumann series}
\REQUIRE Quantum circuit generating the $n$-qubit state $\rho$, \\
              \hskip1.5em the $n$-qubit quantum observable $O$, \\
              \hskip1.5em the $n$-qubit noisy measurement device, \\
              \hskip1.5em noise resistance $\xi$, \\
              \hskip1.5em probability tolerance $\delta$, \\
              \hskip1.5em precision parameter $\varepsilon$.

\ENSURE $\eta$, as an estimate of $\tr[O\rho]$.

\STATE Compute $K =\left\lceil \log\varepsilon/\log\xi - 1\right\rceil$;

\STATE Compute $\Delta = \binom{2K+2}{K+1} - 1$;

\STATE Compute $M =\lceil 2(K+1)\Delta\log(2/\delta)/\varepsilon^2 \rceil$;

\FOR{$k=1,\cdots,K+1$} 
\FOR{$m=1,\cdots,M$} 
\STATE\hskip0.5em Run the quantum circuit to generate $\rho$;

\STATE\hskip0.5em Execute the measurement device $k$ times sequentially;

\STATE\hskip0.5em Obtain the measurement outcome $\bm{s}^{m,k}$;
\ENDFOR
  \STATE Compute $\eta^{(k)} = \frac{1}{M}\sum_{m=1}^{M}O(\bm{s}^{m,k})$;
\ENDFOR

\STATE Compute $\eta = \sum_{k=1}^{K+1} c_K(k-1) \eta^{(k)}$, where $c_K$ is defined in~\eqref{eq:c_k};

\STATE Output $\eta$.

\end{algorithmic}
\end{algorithm}

We claim that the output $\eta$ of Algorithm~\ref{alg:Neumann series}
approximates the expectation value $\tr[O\rho]$ pretty well,
as captured by the following proposition.

\begin{proposition}\label{prop:good-estimate}
The output $\eta$ of Algorithm~\ref{alg:Neumann series} satisfies
\begin{align}\label{eq:good-estimate}
    \opn{Pr}\left\{\vert\tr[O\rho] - \eta\vert\leq 2\varepsilon\right\} \geq 1 - \delta.
\end{align}
\end{proposition}

Proof of the proposition is given in Appendix~\ref{appx:good-estimate}.
Intuitively, Eq.~\eqref{eq:good-estimate} says that the output $\eta$ of
Algorithm~\ref{alg:Neumann series}
estimates the ideal expectation value $\tr[O\rho]$ with error $2\varepsilon$ at a probability
greater than $1-\delta$. Analyzing Algorithm~\ref{alg:Neumann series},
we can see that we ought to expand the Neumann series to the $K$-th order,
where $K$ is computed via~\eqref{eq:K-opt},
and estimate the $K+1$ noisy expectation values $E^{(1)},\cdots,E^{(K+1)}$ individually.
For each expectation, we need $M$ copy of quantum states.
As so, the total number of quantum states consumed is given by
\begin{align}
    M(K+1)
&= 2(K+1)^2\Delta\log(2/\delta)/\varepsilon^2 \nonumber \\
&\approx 4^K\log(2/\delta)/\varepsilon^2.\label{eq:no-of-states}
\end{align}
In other words, our error mitigation method increases the number of
quantum states that is required to achieve the given precision $\varepsilon$ by
a factor of $4^K$ compared with the case of the ideal measurement.
In Fig.~\ref{fig:K_vs_xi}, we plot the optimal truncated number $K$~\eqref{eq:K-opt}
as a function of the noise resistance $\xi$, the error tolerance parameter is fixed as $\varepsilon=0.01$.
One can check from the figure that $K\leq 10$ whenever the noise resistance satisfies $\xi\leq0.657$
(Equivalently, the minimal diagonal element of $A$ is larger than $0.67$).
That is to say, the incurred error mitigation overhead $4^K$
is independent of the system size, so long as the noise resistance $\xi$ is moderate,
in the sense that it is below a certain threshold (say $0.657$).
On the other hand, the number of noisy quantum measurements applied in Algorithm~\ref{alg:Neumann series}
is given by
\begin{align}
    \left(\sum_{k=1}^{K+1}k\right)M \approx 2(K+1)^3\Delta\log(2/\delta)/\varepsilon^2.
\end{align}
Compared to~\eqref{eq:no-of-states}, our method has used more number of measurements than the number
of quantum states by a multiplier $K+1$. We remark that both costs are roughly characterized
by the prominent factor $4^K$.

\begin{figure}
  \centering
  \includegraphics[width=0.48\textwidth]{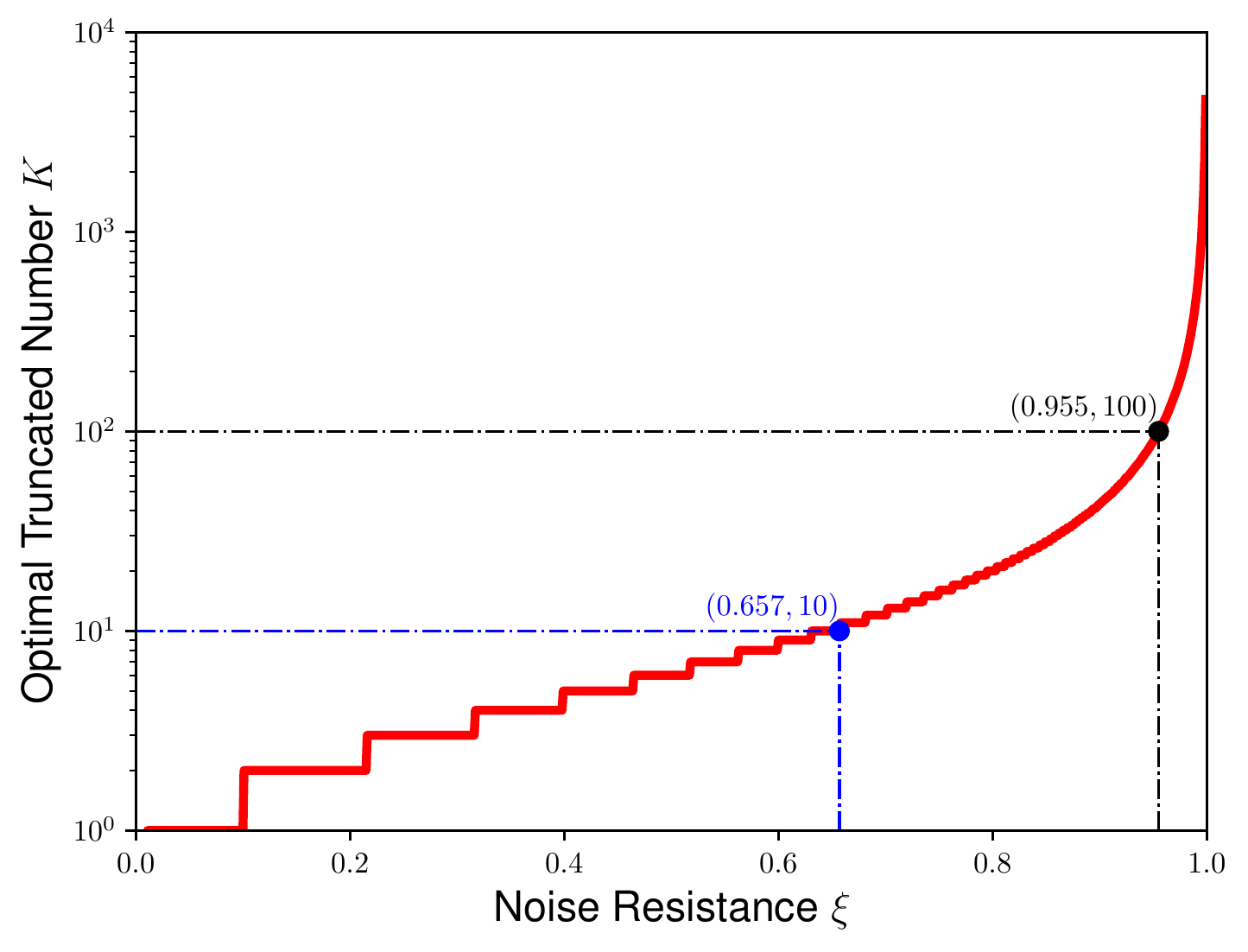}
  \caption{The optimal truncated number $K$~\eqref{eq:K-opt} as a function
          of the noise resistance $\xi$, where the precision parameter is $\varepsilon=0.01$.}
  \label{fig:K_vs_xi}
\end{figure}

\subsection{Discussion on the noise resistance}

In practical applications, $\xi$ can be obtained from the specifications of NISQ devices.
For example, in IBM quantum devices, the specifications are often reported $2-10\%$~\cite{kwon2020hybrid}.
If such information is not available, we may perform calibration to obtain $A$ first
and then compute $\xi$, which is still resource-efficient compared to computing
the inverse matrix $A^{-1}$.

When defining $\xi$ in~\eqref{eq:noise resistance}, we do not consider any structure of $A$.
If certain noise model is assumed, the calculation of $\xi$ can be simplified.
In the following, we consider the tensor product noise model
and show that the noise resistance can be compute analytically.
Assume $A$ is a tensor product of $n$ $2\times 2$ stochastic matrices, i.e.,
\begin{align}
    A_{\opn{tp}} = \begin{bmatrix}
                1 - \alpha_1 & \beta_1 \\
                \alpha_1 & 1 - \beta_1
    \end{bmatrix}
    \otimes
    \cdots
    \otimes
    \begin{bmatrix}
                1 - \alpha_n & \beta_n \\
                \alpha_n & 1 - \beta_n
    \end{bmatrix},
\end{align}
where $\alpha_i$ and $\beta_i$ are error rates describing the $i$-th qubit's readout errors
$0\to1$ and $1\to0$, respectively. One can show that
\begin{align}\label{eq:d-min-TP}
    \xi(A_{\opn{tp}}) = 2\left(1 - \prod_{i=1}^n \min\{1-\alpha_i, 1-\beta_i\}\right).
\end{align}
Specially, if $\alpha_i,\beta_i\ll1$, then $\xi\approx 2(1-1/e^\gamma)$, where
$\gamma := \sum_{i=1}^n\max\{\alpha_i,\beta_i\}$ is called the noise strength in~\cite{bravyi2020mitigating}.


\section{Experimental results}\label{sec:experimental results}

We apply the proposed error mitigation method to the following illustrative example
and demonstrate its performance. Consider the input state $\rho = \proj{\Phi}$, where
\begin{align}
    \ket{\Phi} := \frac{1}{\sqrt{2^n}}\sum_{i=0}^{2^n-1}\ket{i},
\end{align}
which is the maximal superposition state.
The observable $O$ is a tensor product of Pauli $Z$ operators, i.e., $O=Z^{\otimes n}$.
The ideal expectation value is $\tr[O\rho] = 0$. We choose $n=8$ and
randomly generate a noise matrix $A^\ast$ whose noise resistance satisfies
$\xi(A^\ast) \approx 0.657$ (as so the noise matrix is moderate).
We repeat the procedure for producing the noisy expectation
value $\eta^{(1)}$ and Algorithm~\ref{alg:Neumann series} for producing
the mitigated expectation value $\eta$ a total number of $1000$ times.
Note that all these experiments assume the same noise matrix $A$,
and the parameters are chosen as $\varepsilon=\delta=0.01$.
The obtained expectation values are scatted in Fig.~\ref{fig:error-mitigation-result}.
It is easy to see from the figure that the noisy measurement device,
characterized by the noise matrix $A^\ast$, incurs a bias $\approx -0.007$
to the estimated expectation values.
On the other hand, the error mitigated expectation values distributed evenly
around the ideal value $0$ within a distance of $0.01$ with high probability.
As evident from Fig.~\ref{fig:error-mitigation-result}, several
mitigated expectation values fall outside the expected region.
These statistical outcomes match our conclusion in Proposition~\ref{prop:good-estimate},
validating the correctness and performance of the proposed error mitigation method.

Fig.~\ref{fig:error-mitigation-result2} shows the (noisy and
mitigated) expectation values estimated via the above procedure
as a function of the number of qubits. In our numerical setup,
for an experiment whose number of qubits $n$ is less than $8$,
its corresponding noise matrix is obtained by partially tracing
out the rightmost $8-n$ qubit systems from $A^\ast$.
The entire experiment for each $n$ was repeated $1000$ times
in order to estimate the error bars.
The reason that the noisy estimates behave well for single and two qubits
is that the underlying noise matrices are close to the identity in
the total variation distance~\cite{maciejewski2020mitigation}.
It can be seen that the cross-talk
noise presented in the noisy measurement device severely
distorts the estimated expectation value while
our error mitigation method is insensitive to this kind of error.

\begin{figure}[!htbp]
  \centering
  \includegraphics[width=0.48\textwidth]{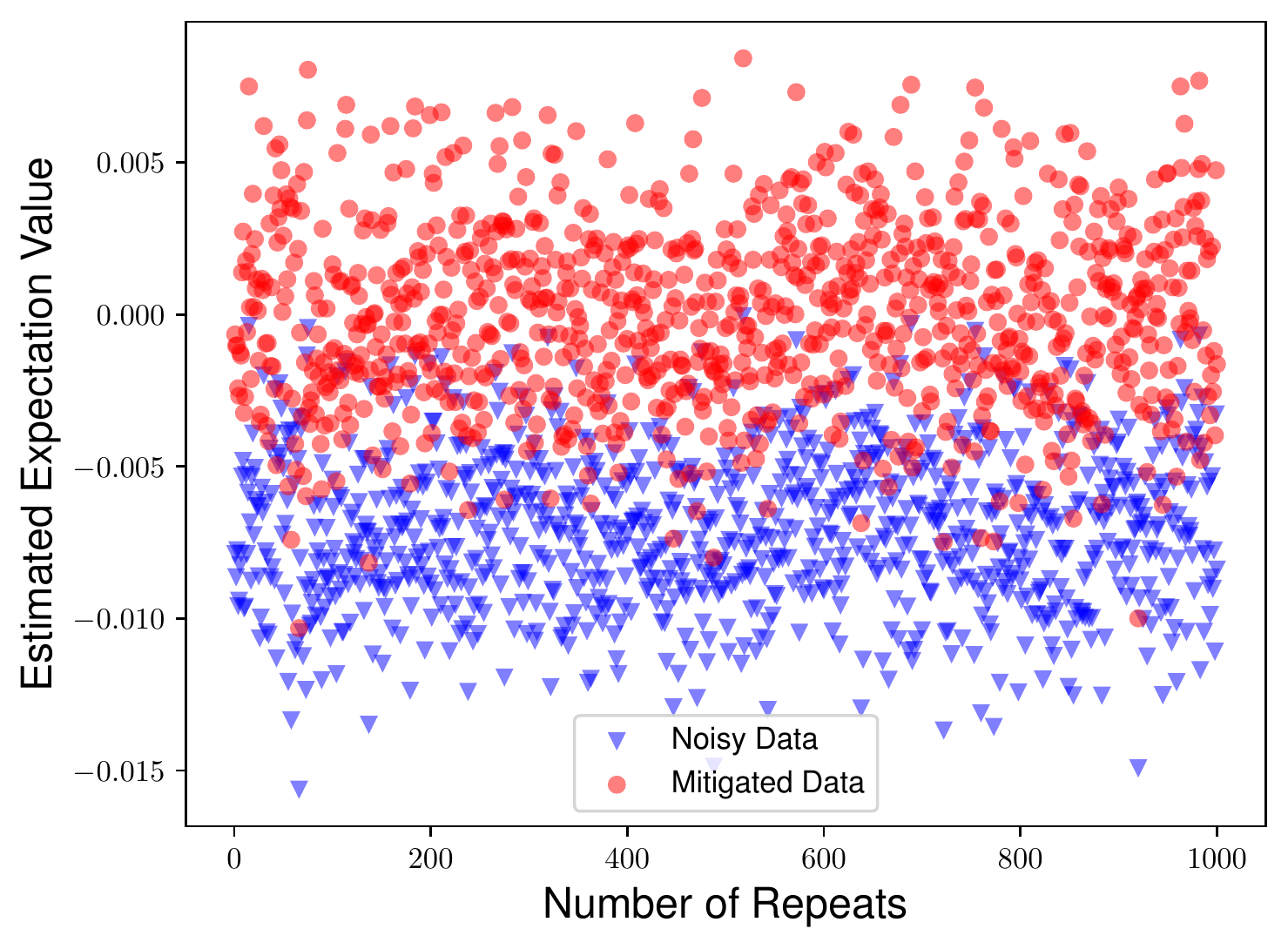}
  \caption{$1000$ noisy estimates $\eta^{(1)}$ (blue triangles)
        and error mitigated estimates $\eta$ (red dots)
        for the ideal expectation value $\tr[O\rho] = 0$.
        Here, the number of qubits is $8$.}
  \label{fig:error-mitigation-result}
\end{figure}

\begin{figure}[!htbp]
  \centering
  \includegraphics[width=0.48\textwidth]{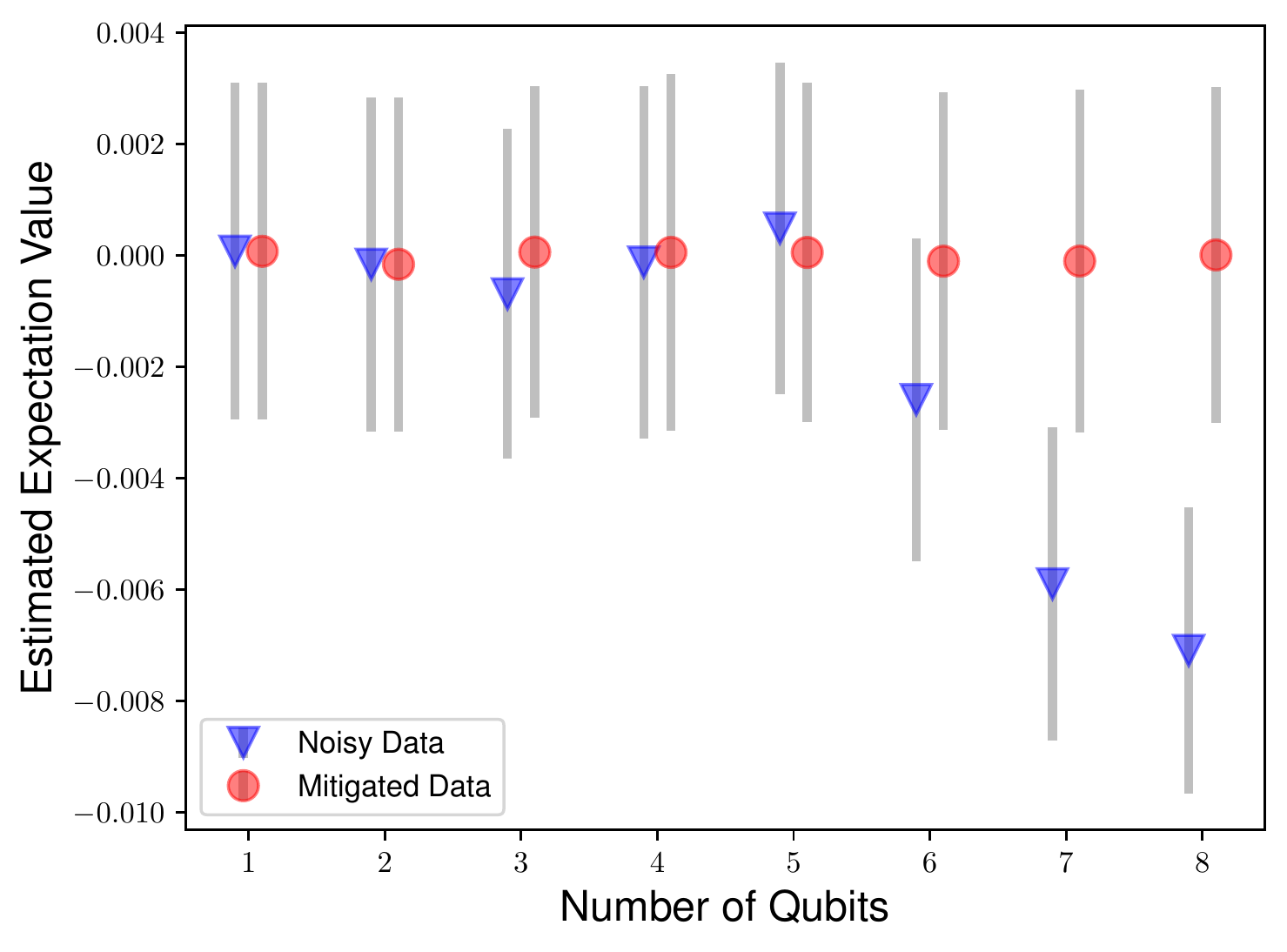}
  \caption{Average expectation values for $1\leq n \leq 8$ qubits obtained
            with (red dots) and without (blue triangles) the error mitigation method.
            Each error bar is estimated by repeating the experiment $1000$ times.}
  \label{fig:error-mitigation-result2}
\end{figure}

\section{Conclusions}\label{sec:conclusions}

We have introduced a scalable method to mitigate measurement errors in computing expectation
values of quantum observables, an essential building block of numerous quantum algorithms.
The idea behind this method is to approximate the inverse of the noise matrix determined
by the noisy measurement device using a small number of the Neumann series terms.
Our method via the truncated Neumann series outperforms the exact matrix inversion 
method by significantly reducing the resource costs in time and samples of quantum 
states while only slightly degrading the error mitigation performance.
In particular, our method works for any classical noise model and
does not require the calibration procedure to learn the noise matrix {a prior}.
Most importantly, the incurred error mitigation overhead is independent of the
system size, as long as the noise resistance of the noisy measurement device is moderate.
This property is beneficial and will be more and more important as the quantum circuit
sizes increase. We have numerically tested this method and found that the computation
accuracy is substantially improved. We believe that the proposed method will be useful
for experimental measurement error mitigation in NISQ quantum devices.

\section*{Acknowledgements}
We thank Runyao Duan for helpful suggestions. 
We would like to thank Zhixin Song for collecting the experiment data.


%

\setcounter{secnumdepth}{2}
\appendix
\widetext 

\section{Proof of Eq.~\eqref{eq:noisy-expectation}}\label{appx:noisy-expectation}

\begin{proof}
By the definition of $\eta^{(1)}$, we have
\begin{align}
  \eta^{(1)} = \frac{1}{M}\sum_{m=1}^MO(\bm{s}^{m,1})
= \frac{1}{M}\sum_{m=1}^M\sum_{\bm{x}\in\{0,1\}^n}O(\bm{x})\langle\bm{x}\vert\bm{s}^{m,1}\rangle
= \sum_{\bm{x}\in\{0,1\}^n}O(\bm{x})\langle \bm{x}\vert
  \left(\frac{1}{M}\sum_{m=1}^M\vert \bm{s}^{m,1}\rangle\right).
\end{align}
The expectation value can be evaluated as
\begin{align}
E^{(1)} := \bE[\eta^{(1)}]
&= \bE\left[\sum_{\bm{x}\in\{0,1\}^n}O(\bm{x})\langle \bm{x}\vert
  \left(\frac{1}{M}\sum_{m=1}^M\vert \bm{s}^{m,1}\rangle\right)\right] \\
&= \sum_{\bm{x}\in\{0,1\}^n}O(\bm{x})\langle \bm{x}\vert\bE\left[\frac{1}{M}\sum_{m=1}^M\vert \bm{s}^{m,1}\rangle\right] \\
&= \sum_{\bm{x}\in\{0,1\}^n}O(\bm{x})\langle \bm{x}\vert A\opn{vec}(\rho).
\end{align}
\end{proof}

\section{Proof of Theorem~\ref{thm:approximation}}\label{appx:approximation}

\begin{proof}
First of all, notice that
\begin{align}
  \left\vert \tr[O\rho] - \sum_{k=1}^{K+1}c_K(k-1) E^{(k)}\right\vert
&=   \left\vert \sum_{\bm{x}\in\{0,1\}^n}O(\bm{x})\langle \bm{x}\vert\opn{vec}(\rho)
  - \sum_{\bm{x}\in\{0,1\}^n}O(\bm{x})\langle \bm{x}\vert
    \left(\sum_{k=1}^{K+1} c_K(k-1)A^k\opn{vec}(\rho)\right)\right\vert \\
&=  \left\vert \sum_{\bm{x}\in\{0,1\}^n}O(\bm{x})\langle \bm{x}\vert\left(
        I - \sum_{k=1}^{K+1} c_K(k-1)A^k\right)\opn{vec}(\rho)\right\vert \\
&= \left\vert \sum_{\bm{x}\in\{0,1\}^n}O(\bm{x})\langle \bm{x}\vert\left(
        I - A\left(\sum_{k=1}^{K+1} c_K(k-1)A^{k-1}\right)\right)\opn{vec}(\rho)\right\vert\\
&= \left\vert \sum_{\bm{x}\in\{0,1\}^n}O(\bm{x})\langle \bm{x}\vert\left(
        I - A\left(\sum_{k=0}^K c_K(k)A^{k}\right)\right)\opn{vec}(\rho)\right\vert \\
&= \left\vert \sum_{\bm{x}\in\{0,1\}^n}O(\bm{x})\langle \bm{x}\vert\left(
        I - A\left(\sum_{k=0}^K (I-A)^k\right)\right)\opn{vec}(\rho)\right\vert\label{eq:appx:approximation-0} \\
&= \left\vert \sum_{\bm{x}\in\{0,1\}^n}O(\bm{x})\langle \bm{x}\vert\left(
        I - \left(I - (I-A)^{K+1}\right)\right)\opn{vec}(\rho)\right\vert\label{eq:appx:approximation-1}  \\
&= \left\vert \sum_{\bm{x}\in\{0,1\}^n}O(\bm{x})\langle \bm{x}\vert(I-A)^{K+1}\opn{vec}(\rho)\right\vert,
    \label{eq:appx:approximation-2}
\end{align}
where~\eqref{eq:appx:approximation-0} follows from~\eqref{eq:Neumann-series} and~\eqref{eq:appx:approximation-1} follows
from the closed-form formula of a geometric series.
Now we show that the quantity in~\eqref{eq:appx:approximation-2} can be bounded from above.
Define the induced matrix $1$-norm of a $m\times n$ matrix $B$ as
\begin{align}
   \norm{B}{1} := \max_{1\leq j \leq n}\sum_{i=1}^n\vert B_{ij}\vert
                \equiv \max_{1\leq j \leq n} \sum_{i=1}^n \vert \bra{i} B \ket{j}\vert,
\end{align}
which is simply the maximum absolute column sum of the matrix.
Let $\rho(\bm{y})$ is the $\bm{y}$-th diagonal element of the quantum state $\rho$.
Consider the following chain of inequalities:
\begin{subequations}\label{eq:bounding-tail}
\begin{align}
\left\vert \sum_{\bm{x}\in\{0,1\}^n}O(\bm{x})\langle \bm{x}\vert(I-A)^{K+1}\opn{vec}(\rho)\right\vert
&= \left\vert\sum_{\bm{x}\in\{0,1\}^n}\sum_{\bm{y}\in\{0,1\}^n}O(\bm{x})
    \rho(\bm{y})\langle \bm{x}\vert (I-A)^{K+1}\vert\bm{y}\rangle\right\vert\\
&\leq\sum_{\bm{x}\in\{0,1\}^n}\sum_{\bm{y}\in\{0,1\}^n}
    \vert O(\bm{x})\vert \cdot
    \rho(\bm{y}) \cdot \left\vert\langle \bm{x}\vert(I-A)^{K+1}\vert\bm{y}\rangle\right\vert\\
&\leq\sum_{\bm{x}\in\{0,1\}^n}\sum_{\bm{y}\in\{0,1\}^n}
    \rho(\bm{y})\left\vert\langle \bm{x}\vert(I-A)^{K+1}\vert\bm{y}\rangle\right\vert\label{eq:appx:approximation-3}\\
&=\sum_{\bm{y}\in\{0,1\}^n}\rho(\bm{y})
  \sum_{\bm{x}\in\{0,1\}^n}\left\vert\langle \bm{x}\vert(I-A)^{K+1}\vert\bm{y}\rangle\right\vert \\
&\leq\sum_{\bm{y}\in\{0,1\}^n} \rho(\bm{y}) \lVert (I-A)^{K+1} \lVert_1\label{eq:appx:approximation-4} \\
&= \lVert (I-A)^{K+1} \lVert_1 \label{eq:appx:approximation-5} \\
&\leq \lVert I-A \lVert_1^{K+1}\label{eq:appx:approximation-6} \\
&= \xi^{K+1},\label{eq:appx:approximation-7}
\end{align}
\end{subequations}
where~\eqref{eq:appx:approximation-3} follows from the assumption
that $\vert O(\bm{x})\vert \leq 1$ (cf. Eq.~\eqref{eq:observable}),
~\eqref{eq:appx:approximation-4} follows from the definition of induced matrix $1$-norm,
~\eqref{eq:appx:approximation-5} follows from the fact that $\rho$ is a quantum state
and thus $\sum_{\bm{y}}\rho(\bm{y})=1$,
~\eqref{eq:appx:approximation-6} follows from the submultiplicativity property of the induced matrix norm,
and~\eqref{eq:appx:approximation-7} follows from Lemma~\ref{lemma:min-norm} stated below.
We are done.
\end{proof}

\begin{lemma}\label{lemma:min-norm}
Let $A$ be a column stochastic matrix of size $d \times d$. It holds that
\begin{align}
    \norm{I - A}{1} = \xi(A),
\end{align}
where $\xi(A)$ is defined in~\eqref{eq:noise resistance}.
\end{lemma}
\begin{proof}
Since $A$ is column stochastic, $I-A$ has non-negative diagonal elements and negative off-negative elements.
Thus
\begin{align}
  \norm{I - A}{1}
&= \max_{1\leq j\leq d}\left(1 - A_{jj} + \sum_{i\neq j} A_{ij}\right) \\
&= \max_{1\leq j\leq d}\left(1 - A_{jj} + 1 - A_{jj}\right) \\
&= 2\max_{1\leq j\leq d}\left(1 - A_{jj}\right) \\
&= 2 - 2\min_{1\leq j\leq d} A_{jj} \\
&=: \xi(A),
\end{align}
where the second line follows from the fact that $A$ is column stochastic.
\end{proof}

\section{Sequential measurements}\label{appx:sequential measurements}

In the Appendix, we prove that the classical noise model describing the
sequential measurement repeating $k$ times is effectively characterized by the
stochastic matrix $A^k$. We begin with the simple case $k=2$.
Since the noise model is classical and linear in the input,
it suffices to consider the computational basis states as inputs.
As shown in Fig.~\ref{fig:sequential-measurement},
we apply the noisy quantum measurement device two times sequentially
on the input state $\proj{\bm{x}}$ in computational basis
where $\bm{x}\in\{0,1\}^{n}$.
Assume the measurement outcome of the first measurement is $\bm{y}$ and
the measurement outcome of the second measurement is $\bm{z}$,
where $\bm{y},\bm{z}\in\{0,1\}^{n}$.
Assume that the noise matrix associated with this sequential measurement is
$A'$. That is, the probability of obtaining the outcome $\bm{z}$ provided
the true outcome is $\bm{x}$ is given by $A'_{\bm{z}\bm{x}}$.
Practically, we input $\proj{\bm{x}}$ to the first noisy measurement device
and obtain the outcome $\bm{y}$. The probability of this event
is $A_{\bm{y}\bm{x}}$, by the definition of the noise matrix.
Similarly, we input $\proj{\bm{y}}$ to the second noisy measurement device
and obtain the outcome $\bm{z}$. The probability of this event
is $A_{\bm{z}\bm{y}}$. Inspecting the chain $\bm{x}\to\bm{y}\to\bm{z}$, we have
\begin{align}
    A'_{\bm{z}\bm{x}}
= \sum_{\bm{y}\in\{0,1\}^{n}} A_{\bm{y}\bm{x}}A_{\bm{z}\bm{y}}
= A^2_{\bm{z}\bm{x}}.
\end{align}
The above analysis justifies that the classical noise model describing the sequential
measurement repeating $2$ times is effectively characterized by the
stochastic matrix $A^2$. The general case can be analyzed similarly.

\begin{figure}[!htbp]
  \centering
  \includegraphics[width=0.3\textwidth]{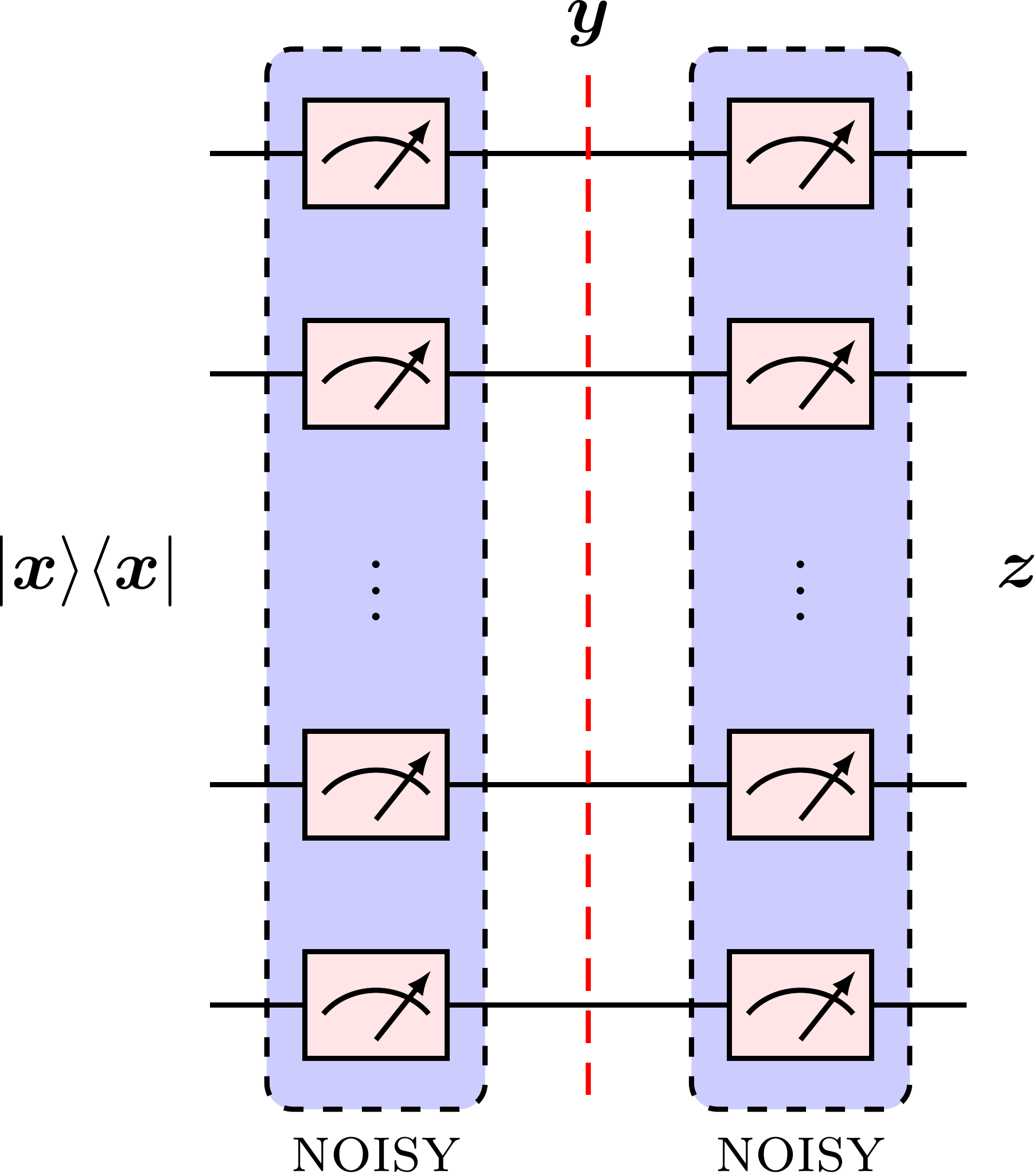}
  \caption{Apply the noisy quantum measurement device two times sequentially
        on the input state $\proj{\bm{x}}$ where $\bm{x}\in\{0,1\}^{n}$.
        The measurement outcome of the first measurement is $\bm{y}$ and
        the measurement outcome of the second measurement is $\bm{z}$.}
  \label{fig:sequential-measurement}
\end{figure}

Mathematically, quantum measurements can be modeled as
quantum-classical quantum channels~\cite[Chapter 4.6.6]{wilde2016quantum}
where they take a quantum system to a classical one.
Experimentally, the implementation of quantum measurement
is platform-dependent and has different characterizations.
For example, the fabrication and control of quantum coherent superconducting circuits
have enabled experiments that implement quantum measurement~\cite{naghiloo2019introduction}.
Based on the outcome data, experimental measurements are typically categorized
into two types: those only output classical outcomes
and those output both classical outcomes and quantum states.
That is, besides the usually classical outcome sequences,
the measurement device will also output a quantum state on the computational basis
corresponding to the classical outcome.
For the former type, we can implement the sequential measurement
via the \emph{qubit reset}~\cite{egger2018pulsed,magnard2018fast,yirka2020qubit} approach,
by which we mean the ability to re-initialize the qubits into a known state,
usually a state in the computational basis, during the course of the computation.
Technically, when the $i$-th noisy measurement outputs an outcome
sequence $\bm{s}^i\in\{0,1\}^n$, we use the qubit reset technique to prepare
the computational basis state $\vert\bm{s}^i\rangle\!\langle\bm{s}^i\vert$ and feed it to
the $(i+1)$-th noisy measurement (cf. Fig.~\ref{fig:sequential-measurement}).
In this case, the noisy measurement device can be reused.
For the latter type, the sequential measurement can be implemented efficiently:
when the $i$-th noisy measurement outputs a classical sequence and a quantum state
on the computational basis, we feed the quantum state to the $(i+1)$-th noisy measurement.

\section{Proof of Proposition~\ref{prop:good-estimate}}\label{appx:good-estimate}

\begin{proof}
By definition,
\begin{align}
\eta = \sum_{k=1}^{K+1} c_K(k-1) \eta^{(k)}
 = \frac{1}{M}\sum_{k=1}^{K+1}\sum_{m=1}^{M}c_K(k-1) O(\bm{s}^{m,k})
 = \frac{1}{M(K+1)}\sum_{k=1}^{K+1}\sum_{m=1}^{M}(K+1)c_K(k-1) O(\bm{s}^{m,k}).
\end{align}
Introducing the new random variables $X_{m,k}:= (K+1)c_K(k-1)O(\bm{s}^{m,k})$, we have
\begin{align}
\eta =  \frac{1}{M(K+1)}\sum_{k=1}^{K+1}\sum_{m=1}^{M} X_{m,k}.\label{eq:TgJvNAKG}
\end{align}
Intuitively, Eq.~\eqref{eq:TgJvNAKG} says that $\eta$ can be viewed as the empirical mean value
of the set of random variables
\begin{align}
    \left\{X_{m,k}: m=1,\cdots,M; k=1,\cdots,K+1\right\}.
\end{align}
First, we show that the absolute value of each $X_{m,k}$ is upper bounded as
\begin{align}\label{eq:mMZgBWZvH1}
  \vert X_{m,k}\vert = \vert (K+1)c_K(k-1)O(\bm{s}^{m,k}) \vert
  \leq (K+1) \vert c_K(k-1)\vert \vert O(\bm{s}^{m,k}) \vert
  \leq (K+1) \vert c_K(k-1)\vert,
\end{align}
where the second inequality follows from the assumption of $O$ (cf. Eq.~\eqref{eq:observable}).
Then, we show that $\eta$ is an unbiased estimator of the quantity $\sum_{k=1}^{K+1}c_K(k-1)E^{(k)}$:
\begin{subequations}\label{eq:mMZgBWZvH2}
\begin{align}
  \bE[\eta]
&= \bE\left[\frac{1}{M(K+1)}\sum_{k=1}^{K+1}\sum_{m=1}^{M} X_{m,k}\right] \\
&= \bE\left[\frac{1}{M}
   \sum_{k=1}^{K+1}\sum_{m=1}^{M}c_K(k-1)O(\bm{s}^{m,k})\right] \\
&= \sum_{k=1}^{K+1} c_K(k-1)\left(\sum_{\bm{x}}O(\bm{x})\langle \bm{x}\vert
      \bE_M\left[ \frac{1}{M} \sum_{m=1}^M \vert \bm{s}^{m,k}\rangle \right]\right) \\
&= \sum_{k=1}^{K+1} c_K(k-1)\left(
    \sum_{\bm{x}}O(\bm{x})\langle \bm{x}\vert A^k\opn{vec}(\rho)\right) \\
&= \sum_{k=1}^{K+1}c_K(k-1) E^{(k)},
\end{align}
\end{subequations}
where the last equality follows from~\eqref{eq:E_k}.
Eqs.~\eqref{eq:mMZgBWZvH1} and~\eqref{eq:mMZgBWZvH2} together guarantee that
the prerequisites of the Hoeffding's inequality hold. By the Hoeffding's equality, we have
\begin{align}
    \Pr\left\{\left\vert\eta - \sum_{k=1}^{K+1}c_K(k-1) E^{(k)} \right\vert \geq \varepsilon\right\}
&\leq 2\exp\left(- \frac{2M^2(K+1)^2\varepsilon^2}{ 4\sum_{k=1}^{K+1}\sum_{m=1}^M((K+1)c_K(k))^2}\right) \\
&= 2\exp\left(- \frac{2M^2(K+1)^2\varepsilon^2}{ 4M(K+1)^3\left(\sum_{k=0}^K [c_K(k)]^2\right)}\right) \\
&= 2\exp\left(- \frac{M\varepsilon^2}{ 2(K+1)\Delta}\right),
\end{align}
where $\Delta:=\sum_{k=1}^{K+1} [c_K(k)]^2 = \binom{2K+2}{K+1} - 1$.
Solving
\begin{align}
    2\exp\left(- \frac{M\varepsilon^2}{ 2(K+1)\Delta}\right) \leq \delta
\end{align}
gives
\begin{align}
    M \geq 2(K+1)\Delta\log(2/\delta)/\varepsilon^2.
\end{align}
To summarize, choosing $K =\left\lceil \log\varepsilon/\log\xi - 1\right\rceil$
and $M =\lceil 2(K+1)\Delta\log(2/\delta)/\varepsilon^2 \rceil$, we are able obtain the following
two statements
\begin{align}
    \Pr\left\{\left\vert\eta - \sum_{k=1}^{K+1}c_K(k-1) E^{(k)} \right\vert \geq \varepsilon\right\} \leq \delta, \\
    \left\vert \tr[O\rho] - \sum_{k=1}^{K+1}c_K(k-1) E^{(k)}\right\vert \leq \varepsilon,
\end{align}
where the first one is shown above and the second one is proved in Theorem~\ref{thm:approximation}.
Using the union bound and the triangle inequality, we conclude that $\eta$
can estimate the ideal expectation value $\tr[O\rho]$ with error $2\varepsilon$ at a probability
greater than $1-\delta$.
\end{proof}

\end{document}